\documentclass[runningheads]{llncs}
\usepackage{amsmath,amssymb,amsfonts}
\usepackage{enumerate,xspace,paralist,tabularx}
\usepackage{subfigure,wrapfig,caption,hyperref}
\usepackage[basic]{complexity}
\usepackage[english]{babel}
\usepackage{todonotes,timetravel}
\usepackage{cite}

\newcommand{\myparagraph}[1]{\medskip\noindent\textbf{\boldmath #1}}
\newcommand{\mysubparagraph}[1]{\smallskip\noindent{#1}}
\setCounters{theorem,lemma}
\let\doendproof\endproof\renewcommand\endproof{~\hfill$\qed$\doendproof}

\newcommand{\FM}{FM-bigraph\xspace}
\newcommand{\FMs}{FM-bigraphs\xspace}
\newcommand{\FMbend}[1]{\textsc{$#1$-bend FM-bigraph}\xspace}
\newcommand{\FMbends}[1]{\textsc{$#1$-bend FM-bigraphs}\xspace}

\newcommand{\Gcell}{G_\mathsf{c}}

\newcommand{\Gcross}{G_\mathsf{x}}
\newcommand{\Gskel}{G_\mathsf{s}}

\newcommand{\cell}{\textrm{cell}}

\newcommand{\remove}[1]{}

\usepackage{nopageno}

\makeatother

\title{Planar Drawings of Fixed-Mobile Bigraphs\thanks{Research in this work started at the Bertinoro Workshop on Graph Drawing 2016. We thank all the participants and in particular S.-H. Hong for useful discussions. We also thank an anonymous reviewer of this research for some valuable comments, and especially for suggesting the idea behind the proof of Theorem~\ref{thm:skeletonhardness}.}}

\author{M.~A.~Bekos\inst{1}, F.~De~Luca\inst{2}, W.~Didimo\inst{2}, T.~Mchedlidze\inst{3}, M.~N\"ollenburg\inst{4}, A.~Symvonis\inst{5}, I.~G.~Tollis\inst{6}}

\authorrunning{M.~A.~Bekos et al.}

\institute{
University of T\"ubingen, Germany, \email{bekos@informatik.uni-tuebingen.de} \and
Universit\`a degli Studi di Perugia, Italy\\\email{felice.deluca@studenti.unipg.it}, \email{walter.didimo@unipg.it} \and
Karlsruhe Institute of Technology, Germany, \email{mched@iti.uka.de} \and
TU Wien, Vienna, Austria, \email{noellenburg@ac.tuwien.ac.at} \and
National Technical University of Athens, Greece, \email{symvonis@math.ntua.gr} \and
University of Crete, Greece, \email{tollis@csd.uoc.gr}
}

\begin{document}

\maketitle

\begin{abstract}
A \emph{fixed-mobile bigraph} $G$ is a bipartite graph such that the vertices of one partition set are given with \emph{fixed} positions in the plane and the \emph{mobile} vertices of the other part, together with the edges, must be added to the drawing.  We assume that $G$ is planar and study the problem of finding, for a given $k \ge 0$, a planar poly-line drawing of $G$ with at most $k$ bends per edge. In the most general case, we show \NP-hardness. For $k=0$ and under additional constraints on the positions of the fixed or mobile vertices, we either prove that the problem is polynomial-time solvable or prove that it belongs to \NP.  Finally, we present a polynomial-time testing algorithm for a certain type of ``layered'' 1-bend drawings.
\end{abstract}

\section{Introduction}\label{se:introduction}
This paper considers the following problem. Let $G=(V_f,V_m,E)$ be a planar bipartite graph such that the vertices in $V_f$, called \emph{fixed vertices}, have fixed distinct locations (points) in the plane, while the vertices in $V_m$, called \emph{mobile vertices}, can be freely placed. Does $G$ admit a crossing-free drawing~$\Gamma$ with at most $k$ bends per edge, where $k$ is a given non-negative integer? We assume that each vertex of $G$ is drawn in $\Gamma$ as a distinct point of the plane and that each edge is drawn as a simple poly-line. We call $G$ an \emph{\FM} and a drawing $\Gamma$ with the properties mentioned above a \emph{planar $k$-bend drawing} of $G$. In particular, since edge bends negatively affect the readability of a graph layout (see, e.g.,~\cite{DBLP:journals/vlc/Purchase02,DBLP:journals/ese/PurchaseCA02}), we are mainly interested in drawings with small values of $k$, ideally with $k=0$ (i.e., straight-line drawings). We define the \emph{bend number} of $G$ as the minimum value of $k$ for which $G$ admits a planar $k$-bend drawing.

Besides its intrinsic theoretical interest, our problem is motivated by the following practical scenario. Fixed vertices represent geographic locations and each mobile vertex is an attribute of one or more locations. One wants to place each mobile vertex in the plane and connect it to its associated locations, while guaranteeing a ``readable'' layout.  We interpret readability in terms of planarity and small number of bends per edge. Other criteria can also be studied, like for example angular resolution, edge length, drawing area (see, e.g.,~\cite{DBLP:books/ph/BattistaETT99,DBLP:reference/crc/2013gd}).

\myparagraph{Contribution.} We introduce $k$-bend drawings of \FMs with a focus on $k=0$ and $k=1$. Our results are as follows:

\mysubparagraph{$(i)$} We prove that computing the bend number of an \FM $G$ is $\NP$-hard. More generally,
deciding whether $G$ admits a planar $k$-bend drawing is at least as hard as deciding whether an $n$-vertex graph admits a planar embedding with given correspondence (mapping) on a set of $n$ points, such that each edge has at most $2k+1$ bends (Section~\ref{sse:hardness}). If all fixed vertices of $G$ are collinear, the existence of a 0-bend drawing can be tested in linear time.

\mysubparagraph{$(ii)$} Since it is difficult to discretize the problem in the general case~\cite{DBLP:journals/dcg/GoaocKOSSW09,DBLP:journals/ijfcs/Patrignani06}, we investigate the case in which each mobile vertex is restricted to lie in the convex hull of its neighbors. This scenario is reasonable in practice, as the user may expect that each attribute is placed in a sort of ``barycentric'' position with respect to its associated locations. In this setting, we prove that testing the existence of a 0-bend drawing is a problem in {\NP}. With a reduction to a combinatorial problem, which is of its own independent interest (but unfortunately \NP-hard in its general form), we obtain polynomial-time solutions when the intersection graph of the convex hulls is a path, a cycle or, more generally, a cactus (Section~\ref{sse:ch-internal}).

\mysubparagraph{$(iii)$} We finally study 1-bend drawings of \FMs in a convention called \emph{$h$-strip drawing model}, inspired by practical labeling scenarios~\cite{DBLP:journals/jgaa/BekosCFH0NRS15}. All fixed vertices are partitioned into a finite set of horizontal strips and each mobile vertex is placed outside these strips. Edges are not allowed to cross any strip, i.e., to intersect both its top and bottom side; see also Fig.~\ref{fi:strip-model}. For this model we provide polynomial-time testing algorithms (Section~\ref{se:1-bend}).

\myparagraph{Related Work.} Our problem is related to several problems addressed in the literature, but it also has substantial differences from all of them.

\mysubparagraph{\em Point labeling.} A close connection is with the problem of labeling a given set of points in the plane (see, e.g.,~\cite{DBLP:conf/dagstuhl/Neyer99,Wolff-96}), because mobile vertices can be regarded as labels for the fixed vertices (points). Similarly to our setting, the \emph{many-to-one boundary labeling} problem~\cite{DBLP:journals/jgaa/BekosCFH0NRS15,l-cmblwh-10} assumes that each label can have multiple associated vertices and it is visually connected to them by poly-line edges. However, edges can only be drawn as chains of horizontal and vertical segments (which may partially overlap), and the labels are placed outside a single rectangular region that encloses all vertices. Variants of the boundary labeling problem, where each fixed vertex is associated with exactly one label are also studied in the literature (see, e.g.,~\cite{DBLP:journals/comgeo/BekosKSW07,DBLP:journals/algorithmica/KindermannNRS0W16,bgnn-rbl-15}). Note that, in labeling problems, labels are geometric shapes of non-empty area, while we model mobile vertices as points.

\mysubparagraph{\em Partial drawings.} Our problem is a special case of the problem of extending a partial drawing of a (not necessarily bipartite) planar graph $G$ to a planar straight-line drawing of $G$. This problem is $\NP$-hard in general~\cite{DBLP:journals/ijfcs/Patrignani06} and polynomial-time solvable for restricted cases~\cite{t-hdg-63,dgk-pdhgg-11,cegl-dgppo-12,HongN08,MchedlidzeNR16}.

\mysubparagraph{\em Point-set embedding.} In a point-set embedding problem, a planar graph with $n$ vertices must be planarly mapped onto a given set of $n$ points, with or without a predefined correspondence between the vertices and the points (see, e.g.,~\cite{DBLP:journals/tcs/BadentGL08,DBLP:journals/algorithmica/BrandesEEFFGGHKKa11,DBLP:journals/jgaa/GiacomoDLMTW08,DBLP:journals/algorithmica/GiacomoLT10,DBLP:journals/jgaa/KaufmannW02,DBLP:journals/gc/PachW01}). Thus, in all settings of the point-set embedding problem, each vertex can only be mapped to a finite set of points. The results in~\cite{DBLP:journals/tcs/BadentGL08,DBLP:journals/gc/PachW01} imply that any $n$-vertex planar \FM admits a $k$-bend drawing with $k = O(n)$. Indeed,~\cite{DBLP:journals/tcs/BadentGL08,DBLP:journals/gc/PachW01} prove that any $n$-vertex planar graph can be planarly mapped onto any set of $n$ points, with given correspondences, using a linear number of bends per edge (which is also necessary in some cases). Hence, for a given \FM, one can place the mobile vertices anywhere so to realize a planar drawing.

\mysubparagraph{\em Constrained drawings of bipartite graphs.} Misue~\cite{DBLP:journals/ieicet/Misue08} proposed a model and a technique for drawing bipartite graphs such that the vertices of a partition set, called \emph{anchors}, are evenly distributed on a circle. Anchors are similar to fixed vertices in our setting, but the order of the anchors in Misue's model can be freely chosen. Extensions to the 3D space and to semi-bipartite graphs have been subsequently presented~\cite{DBLP:conf/hci/ItoMT09,DBLP:conf/iv/MisueZ11}. Finally, several papers study how to draw a bipartite graph such that the vertices of each partition set are on a line or within a specific plane region (see, e.g.,~\cite{DBLP:conf/compgeom/Biedl98,DBLP:conf/wg/BiedlKM98,DBLP:conf/ciac/FossmeierK97}). In these scenarios, the vertices do not have predefined~locations.

\myparagraph{Notation.}
We assume familiarity with graph theory (see, e.g.,~\cite{h-gt-72}). For standard definitions on \emph{planar graphs} and \emph{drawings}, we point the reader to~\cite{DBLP:conf/dagstuhl/1999dg,DBLP:reference/cg/TamassiaL04}.

We denote an \FM by a pair $\langle G, \phi \rangle$, where $G=(V_f,V_m,E)$ is a bipartite graph and $\phi : V_f \rightarrow \mathbb{R}^2$ is a function that maps each vertex $v \in V_f$ to a distinct point $p_v=\phi(v)$.    
A \emph{$k$-bend drawing} of $\langle G, \phi \rangle$ is a $k$-bend drawing of $G$ such that each vertex $v \in V_f$ is mapped to $\phi(v)$. 
In order to study the complexity of computing the bend number of planar \FMs, we introduce the \FMbend{k} decision problem: \emph{Given a planar \FM $\langle G, \phi \rangle$ and a non-negative integer $k$, is there a planar $k$-bend drawing of $\langle G, \phi \rangle$?}

From now on, we assume that $G$ is planar. Also, we let $n_f=|V_f|$, $n_m=|V_m|$, and $n=n_f+n_m$.
Some proofs are moved to the appendix.


\section{Straight-line Planar Drawings of \FMs}\label{se:spd}
We show that the \FMbend{0} problem is $\NP$-hard (Theorem~\ref{th:hardness}), which implies that it is $\NP$-hard to compute the bend number of a planar \FM. A simple linear-time testing algorithm is given when all fixed vertices are collinear (Theorem~\ref{th:collinear}). If each mobile vertex must be placed inside the convex hull of its neighbors, then the \FMbend{0} problem belongs to $\NP$ (Theorem~\ref{th:np}), and it becomes polynomial-time solvable if the intersection graph of the convex hulls is a cactus (Theorem~\ref{thm:cacti}).

\subsection{$\NP$-hardness and Collinear Fixed Vertices}\label{sse:hardness}
To prove that \FMbend{0} is $\NP$-hard we use a reduction from the 1-bend point set embeddability with correspondence problem (or 1-BPSEWC, for short), which has been proven to be $\NP$-hard by Goaoc \emph{et al.}~\cite{DBLP:journals/dcg/GoaocKOSSW09}. Problem~1-BPSEWC is defined as follows: \emph{Given a planar graph $G=(V,E)$, a set $S$ of $|V|$ points in the plane, and a one-to-one correspondence $\zeta$ between $V$ and $S$, is there a planar $1$-bend drawing of $G$ such that each vertex $v$ is mapped to point $\zeta(v)$?}

\begin{theorem}\label{th:hardness}
The \FMbend{0} problem is $\NP$-hard, even if each mobile vertex has degree at most two.
\end{theorem}
\begin{proof}
Let $\langle G=(V,E),S,\zeta \rangle$ be an instance of 1-BPSEWC.
Construct (in linear time) an instance $\langle G'=(V_f,V_m,E'), \phi \rangle$ of \FMbend{0} as follows: Let $V_f = V$ and $\phi=\zeta$; for each edge $e=(u,v) \in E$, define a corresponding vertex $w_e \in V_m$ and two edges $(w_e,u)$, $(w_e,v)$ in $E'$. Clearly, $G$ has a 1-bend drawing $\Gamma$ that respects $\zeta$ if and only if $G'$ has a planar 0-bend drawing $\Gamma'$ that respects $\phi$: The position of a bend along an edge $e=(u,v)$ of $\Gamma$ corresponds to the positions of the mobile vertex $w_e$ in $\Gamma'$; if $e$ has no bend, $w_e$ is drawn anywhere along segment $\overline{uv}$.
\end{proof}

The reduction in Theorem~\ref{th:hardness} can be applied with no change to prove that, for any $k \geq 0$, problem \FMbend{k} is at least as difficult as problem $(2k+1)$-BPSEWC, which allows up to $(2k+1)$ bends per edge.

\begin{theorem}\label{th:k-bend-complex}
The \FMbend{k} problem is  at least as hard as the \mbox{$(2k+1)$-BPSEWC} problem, for any $k\geq 0$.
\end{theorem}

When all fixed vertices of an \FM $\langle G, \phi \rangle$ are~collinear, it can be checked in linear time whether $\langle G, \phi \rangle$ admits a planar 0-bend~drawing.

\newcommand{\collinear}{Let $\langle G=(V_f,V_m,E),\phi \rangle$ be an $n$-vertex \FM such that all vertices of $V_f$ are collinear. There exists an $O(n)$-time algorithm that tests whether $\langle G, \phi \rangle$ admits a planar 0-bend drawing.}
\wormhole{collinear}
\begin{theorem}\label{th:collinear}
\collinear
\end{theorem}
\begin{proof}[sketch]
Let $\ell$ be the line passing through all fixed vertices. Deciding whether $\langle G, \phi \rangle$ has a planar 0-bend drawing coincides with testing the planarity of a graph obtained by augmenting $G$ with a cycle that connects all fixed vertices in the order they appear along $\ell$.
\end{proof}

\subsection{Mobile Vertices at Internal Positions}\label{sse:ch-internal}
We now focus on \emph{convex-hull drawings}, in which all fixed vertices are in general position and each vertex $u_m \in V_m$ lies in the convex hull of its neighbors. With slight abuse of notation, we denote by $CH(u_m)$ the convex hull of the neighbors of $u_m$. Let $\mathcal A = \mathcal A(V_f)$ be the arrangement of lines defined by all pairs of fixed points; see Fig.~\ref{fig:graphs}(a). $\mathcal A$ has $O(n_f^2)$ lines and $O(n_f^4)$ cells~\cite{h-a-04}. Lemma~\ref{lem:cells_equiv} allows us to discretize the set of possible positions for the mobile vertices; it implies that all positions of $u_m$ in the same cell of $\mathcal A$ within $CH(u_m)$ are equivalent for a planar 0-bend drawing of $\langle G, \phi \rangle$.

\newcommand{\cellsequiv}{Let $\langle G=(V_f,V_m,E), \phi \rangle$ be an \FM, $u_m \in V_m$, and $C$ a cell of $\mathcal A=A(V_f)$ inside $CH(u_m)$. Let also $p$ and  $p'$ be two points in~$C$. Suppose that $\Gamma$ is a $0$-bend drawing of $\langle G, \phi \rangle$ where $u_m$ is at point $p$ and let $\Gamma'$ be a $0$-bend drawing of $\langle G, \phi \rangle$ obtained from $\Gamma$ by only moving $u_m$ from point $p$ to point $p'$. Then $\Gamma'$ is planar if and only if $\Gamma$ is planar.}
\wormhole{cellequiv}
\begin{lemma}\label{lem:cells_equiv}
	\cellsequiv
\end{lemma}
\begin{proof}
Suppose by contradiction that $\Gamma'$ is planar and $\Gamma$ is not (the proof for the other direction is symmetric). This implies that while moving $u_m$ along some trajectory $T$ from $p'$ to $p$ inside cell $C$, at some point we get a crossing along one of the edges incident to $u_m$. Let $r$ be the point of $T$ closest to $p'$, such  that placing $u_m$ at $r$ causes such a crossing (i.e., placing $u_m$ on any point between $r$ and $p'$ implies no crossing).
Let $(u_m,u_f)$ be an edge crossed by some other edge $(w_m,w_f)$, when $u_m$ is placed at $r$. Assume w.l.o.g. that $w_m \in V_m$ and $w_f \in V_f$, and that $u_m$ lies to the right of the oriented edge $(w_m,w_f)$; see Fig.~\ref{fi:crossing}. Denote by $\ell(u_f,w_f)$ the line through $u_f$ and $w_f$ and by $\ell(u_f,w_m)$ the line through $u_f$ and $w_m$. Let $\cal R$ be the region delimited by lines $\ell(u_f,w_f)$, $\ell(u_f,w_m)$ and edge $(w_m,w_f)$ that contains $u_m$. Let $r'$ be a point of $T$ lying between $r$ and $p'$. Notice that $r'$ has to lie outside of $\cal R$. Thus $T$ crosses the border of $\cal R$. Let us additionally assume that $r'$ lies arbitrarily close to the border of $\cal R$. We distinguish three cases, based on whether $T$ crosses $\ell(u_f,w_f)$, $\ell(u_f,w_m)$, or edge $(w_m,w_f)$.

\begin{figure}[t]
	\centering
	\subfigure[]{\label{fi:crossing}\includegraphics[scale=0.8]{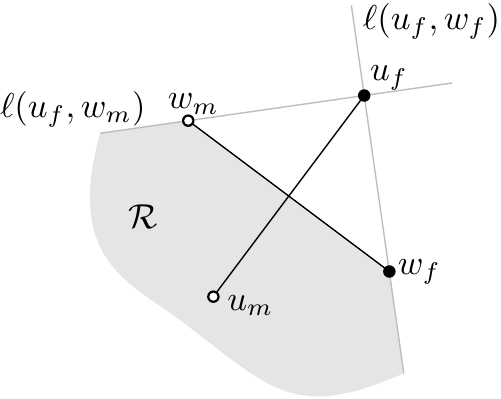}} 
	\hfil
	\subfigure[]{\label{fi:crossing_over}\includegraphics[scale=0.8]{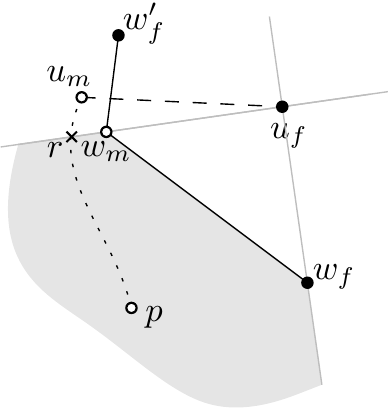}}
	\hfil
	\subfigure[]{\label{fi:crossing_through}\includegraphics[scale=0.8]{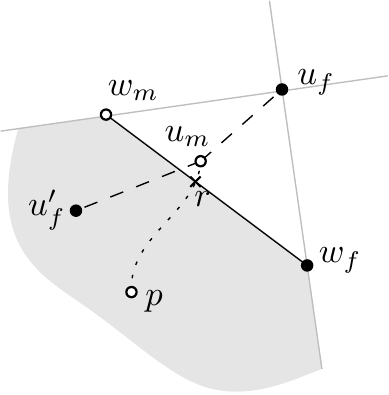}}
	\caption{Illustration of the proof of Lemma~\ref{lem:cells_equiv}. (b-c) Part of trajectory $T$ between $r$ and $r'$ is shown by a dotted line.}
\end{figure}

\smallskip\noindent{\bf Case 1.} $T$ crosses $\ell(u_f,w_f)$. Since line $\ell(u_f,w_f)$ is part of $\mathcal A$, $r'$ lies outside $C$. This is a contradiction to the assumption that $T$ lies in $C$.

\smallskip\noindent{\bf Case 2.} $T$ crosses $\ell(u_f,w_m)$; see Fig.~\ref{fi:crossing_over}. Since $w_m$ is in the convex hull of its neighbors, there is an edge $(w_m,w_f')$ with $w_f$ and $w_f'$ on different sides of $\ell(u_f,w_m)$ (if not, the crossing is resolved). Placing $u_m$ at $r'$ yields a crossing with $(w_m,w_f')$, as $r'$ is arbitrarily close to $\ell(u_f,w_m)$; a contradiction to the choice of~$r$.

\smallskip\noindent{\bf Case 3.} $T$ crosses $(w_m,w_f)$. Since $u_m$ lies in the convex hull of its neighbors, there is an edge $(u_m,u_f')$, where $u_f'$ and $u_f$ lie on different sides of the line through edge $(w_f,w_m)$; see Fig.~\ref{fi:crossing_through}. Placing $u_m$ at $r'$ would introduce a crossing between $(u_m,u_f')$ and $(w_m,w_f)$, as $r'$ lies arbitrarily close to $(w_m,w_f)$. This again contradicts the choice of $r$.
\end{proof}

Lemma~\ref{lem:cells_equiv} implies that, the \FMbend{0} problem belongs to $\NP$ for convex-hull drawings\footnote{We remark that in a preliminary version of~\cite{DBLP:journals/ijfcs/Patrignani06}, it is claimed membership in $\NP$ for the partial planarity extension problem~\cite{Patrignani06}, which would imply membership in $\NP$ also for our problem. That claim, however, lacks a proof in~\cite{Patrignani06} and the author was only able to prove the $\NP$-hardness of the problem in~\cite{DBLP:journals/ijfcs/Patrignani06} (personal communication).}. A non-deterministic algorithm guesses an assignments of the mobile vertices to the $O(n_f^4)$ cells and, since $G$ is planar, checks in $O(n_f^2)$ time whether the corresponding 0-bend drawing is planar (note that $n_m=O(n_f)$). We summarize this observation in the following theorem.

\begin{theorem}\label{th:np}
The \FMbend{0} problem belongs to $\NP$ if each mobile vertex must lie in the convex hull of its neighbors.
\end{theorem}

Central ingredients to prove that the problem is in fact in $\P$ for certain input configurations are
the CH intersection graph $\Gcross$ of $G$, the cell graph $\Gcell$ of $G$, and the skeleton graph $\Gskel$ of $\Gcell$, which we formally define in the following.

The \emph{CH intersection graph} $\Gcross$ is defined as the intersection graph~\cite{McKee1672910} of the convex hulls $CH(u)$ over all $u \in V_m$.

\begin{figure}[tbp]
	\centering
		\includegraphics[scale=1]{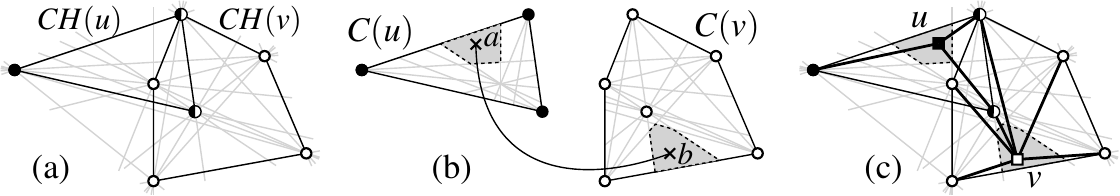}
	\caption{(a) Line arrangement $\mathcal A = \mathcal A(V_f)$ with the neighbors $N(u)$ in black and $N(v)$ in white; $CH(u)$ and $CH(v)$ intersect and thus form an edge in $\Gcross$. (b) Two clusters $C(u)$ and $C(v)$ of $\Gcell$ with one exemplary edge between two cell vertices $a$ and $b$. (c) Placing $u$ inside $\cell(a)$ and $v$ inside $\cell(b)$ yields a planar drawing of the \FM (thick edges).}
	\label{fig:graphs}
\end{figure}

The \emph{cell graph} $\Gcell$ is a clustered graph defined as follows; see Fig.~\ref{fig:graphs}(b) for an example. Each mobile vertex $u$ is associated with a cluster $C(u)$; the vertices of $C(u)$, called  \emph{cell vertices}, are the cells of ${\cal A}$ that intersect with $CH(u)$ (and in fact are contained in $CH(u)$). The vertices of $\Gcell$ are defined by the disjoint union of the vertices of all clusters\footnote{Cells in the intersection of two convex hulls correspond to different vertices of $G_c$.}, that is, $V(G_c) = \uplus_{u\in V_m}{C(u)}$. For a cell vertex $a$ of $\Gcell$, we denote by $\cell(a)$ the cell corresponding to $a$ in ${\cal A}$. For a pair of mobile vertices $u$ and $v$ such that $CH(u) \cap CH(v) \neq \emptyset$, a cell vertex $a \in C(u)$ is adjacent to a cell vertex $b \in C(v)$ if and only if placing $u$ in $\cell(a)$ and $v$ in $\cell(b)$ produces no crossing among the edges incident to $u$ and $v$; see Fig.~\ref{fig:graphs}(c). Note that $\Gcell$ has $O(n_f^4n_m)$ vertices and $O(n_f^8n_m^2)$ edges. Also, by definition, for each pair of mobile vertices $u$ and $v$ such that $CH(u) \cap CH(v) \neq \emptyset$, $u$ and $v$ can be positioned within their convex hulls without creating edge crossings if and only if there exist two adjacent cell vertices $a \in C(u)$ and $b \in C(v)$ in $\Gcell$.

The \emph{skeleton graph} $\Gskel$ is created by selecting exactly one cell vertex, called a \emph{skeleton vertex}, from each cluster of $\Gcell$, such that for every pair of mobile vertices $u$ and $v$ with $CH(u) \cap CH(v) \neq \emptyset$, the skeleton vertices of~$C(u)$ and $C(v)$ are adjacent in $\Gcell$. Graph $\Gskel$ is the subgraph of $G_c$ induced by the skeleton vertices. Note that $\Gskel$ might not exist. If $\Gskel$ exists, then it is isomorphic to $\Gcross$. The following characterization is an immediate consequence of our definitions. 

\newcommand{\fromgeometrytotopology}{An \FM $\langle G,\phi \rangle$ admits a planar 0-bend convex-hull drawing if and only if cell graph $\Gcell$ has a skeleton.}
\wormhole{fromgeometrytotopology}
\begin{lemma}\label{lem:fromgeometrytotopology}
	\fromgeometrytotopology
\end{lemma}
\begin{proof}
A planar 0-bend drawing immediately defines a skeleton. Conversely, if $\Gcell$ has a skeleton $\Gskel$, a planar 0-bend drawing $\Gamma$ is obtained by placing each $u \in V_m$ in the cell corresponding to the skeleton vertex of $C(u)$ in $\Gskel$. Since crossings may only occur between edges incident to mobile vertices $u$ and $v$ such that $CH(u) \cap CH(v) \neq \emptyset$, $\Gamma$ is planar.
\end{proof}

The characterization of Lemma~\ref{lem:fromgeometrytotopology} allows us to translate the geometric problem of finding a 0-bend convex-hull drawing of an \FM bigraph $\langle G,\phi \rangle$ to a purely combinatorial problem on a support clustered graph $\Gcell$ constructed from $\langle G,\phi \rangle$. Unfortunately, however, this combinatorial problem is $\NP$-hard in its general form, as Theorem~\ref{thm:skeletonhardness} shows. Nonetheless, we are able to solve it efficiently when $\Gcross$ is a cactus (Theorem~\ref{thm:cacti}), which includes the special cases in which $\Gcross$ is a cycle or a tree. The next two lemmas are base cases for Theorem~\ref{thm:cacti}.

\begin{lemma}\label{lem:path}
Let  $\langle G,\phi \rangle$ be an \FM such that $\Gcross$ is a path. There exists a polynomial-time algorithm that tests whether $\langle G, \phi \rangle$ has a planar 0-bend convex-hull drawing.
\end{lemma}
\begin{proof}
By Lemma~\ref{lem:fromgeometrytotopology}, it is enough to test whether $\Gcell$ has a skeleton. Let $u_1,\dots, u_\lambda$ be the mobile vertices in the order their convex hulls appear along path $\Gcross$. Call a cell vertex $a$ of $C(u_i)$ \emph{active} if and only if the subgraph of $\Gcell$ induced by $C(u_1)\cup \dots \cup C(u_i)$ has a skeleton containing $a$, where $1 \leq i \leq \lambda$. Thus, $\Gcell$ has a skeleton if and only if there is an active cell vertex in $C(u_\lambda)$.
A simple algorithm that tests this condition works as follows. Initially mark all cell vertices of $C(u_1)$ as active, and then \emph{propagate this information forward} to the cell vertices of $C(u_\lambda)$, that is, for each $i =2,\dots,\lambda$, mark each cell vertex of $C(u_i)$ as active if it has an active neighbor in $C(u_{i-1})$. The time complexity is bounded by the number of vertices and edges in~$\Gcell$.
\end{proof}

\begin{lemma}\label{lem:cycle}
Let  $\langle G,\phi \rangle$ be an \FM such that $\Gcross$ is a simple cycle. There exists a polynomial-time algorithm that tests whether $\langle G, \phi \rangle$ has a planar 0-bend convex-hull drawing.
\end{lemma}
\begin{proof}
Let $u_1,\dots, u_\lambda$ be the mobile vertices in the cyclic order their convex hulls appear along $\Gcross$. Our approach is similar to the one of Lemma~\ref{lem:path}, but now it is not enough to propagate the information about the active vertices only from $C(u_1)$ to $C(u_\lambda)$. Indeed, there might be vertices of $C(u_1)$ that cannot close a cycle with an active vertex of $C(u_\lambda)$.

In the first phase, our refined algorithm starts by marking all cell vertices of $C(u_1)$ as active and then propagates this information forward to $C(u_\lambda)$, as in the case of a path. However, all cell vertices of a cluster that are not marked as active at the end of this phase are now definitely removed from $\Gcell$ (along with their incident edges), as they cannot occur in any skeleton. Then, the algorithm cleans all vertex marks and executes a \emph{backward propagation} phase from $C(u_\lambda)$ to $C(u_1)$ (symmetric to the previous one), where all the remaining vertices in $C(u_\lambda)$ are initially marked as active. As before, all vertices that are not marked as active at the end of this phase are definitely removed from $\Gcell$. Now, the algorithm removes from $C(u_1)$ all vertices with no neighbor in $C(u_\lambda)$ and from $C(u_\lambda)$ all vertices with no neighbor in $C(u_1)$, as these vertices cannot occur in a skeleton of $\Gcell$. Finally, for each pair of adjacent vertices $v \in C(u_1)$ and $w \in C(u_\lambda)$ in $\Gcell$, the algorithm checks whether there exists a path $\pi_{vw}$ from $v$ to $w$ that passes through each $C(u_j)$ exactly once $(j=1, \dots, \lambda)$.
If $\pi_{vw}$ exists, both $v$ and $w$ are marked as \emph{confirmed}. At the end, every vertex in $C(u_1) \cup C(u_\lambda)$ that is not confirmed is removed from $\Gcell$, as it cannot occur in a skeleton. Conversely, by construction, every remaining vertex $v$ in $C(u_1)$ has an adjacent vertex $w \in C(u_\lambda)$ such that $\pi_{vw} \cup (v,w)$ is a skeleton (simple cycle) of $\Gcell$. Thus the test is positive if and only if $C(u_1)$ is not empty.

It is immediate to see that the whole testing algorithm works in polynomial time in the size of $\Gcell$ and, if the test is positive, a skeleton of $\Gcell$ can be easily reconstructed visiting $\Gcell$ from any vertex $v \in C(u_1)$.
\end{proof}

\noindent We now extend the previous result to the case that $\Gcross$ is a cactus, which also covers the case of a tree. We recall that a cactus is a connected graph in which any two simple cycles share at most one vertex. A cactus is an outerplanar graph and can always be decomposed into a tree where each node corresponds to either a single vertex or a simple cycle (refer to Fig.~\ref{fi:cactus-a}).

\begin{figure}[b!]
	\centering
	\subfigure[$\Gcross$]{\label{fi:cactus-a}\includegraphics[height=5cm]{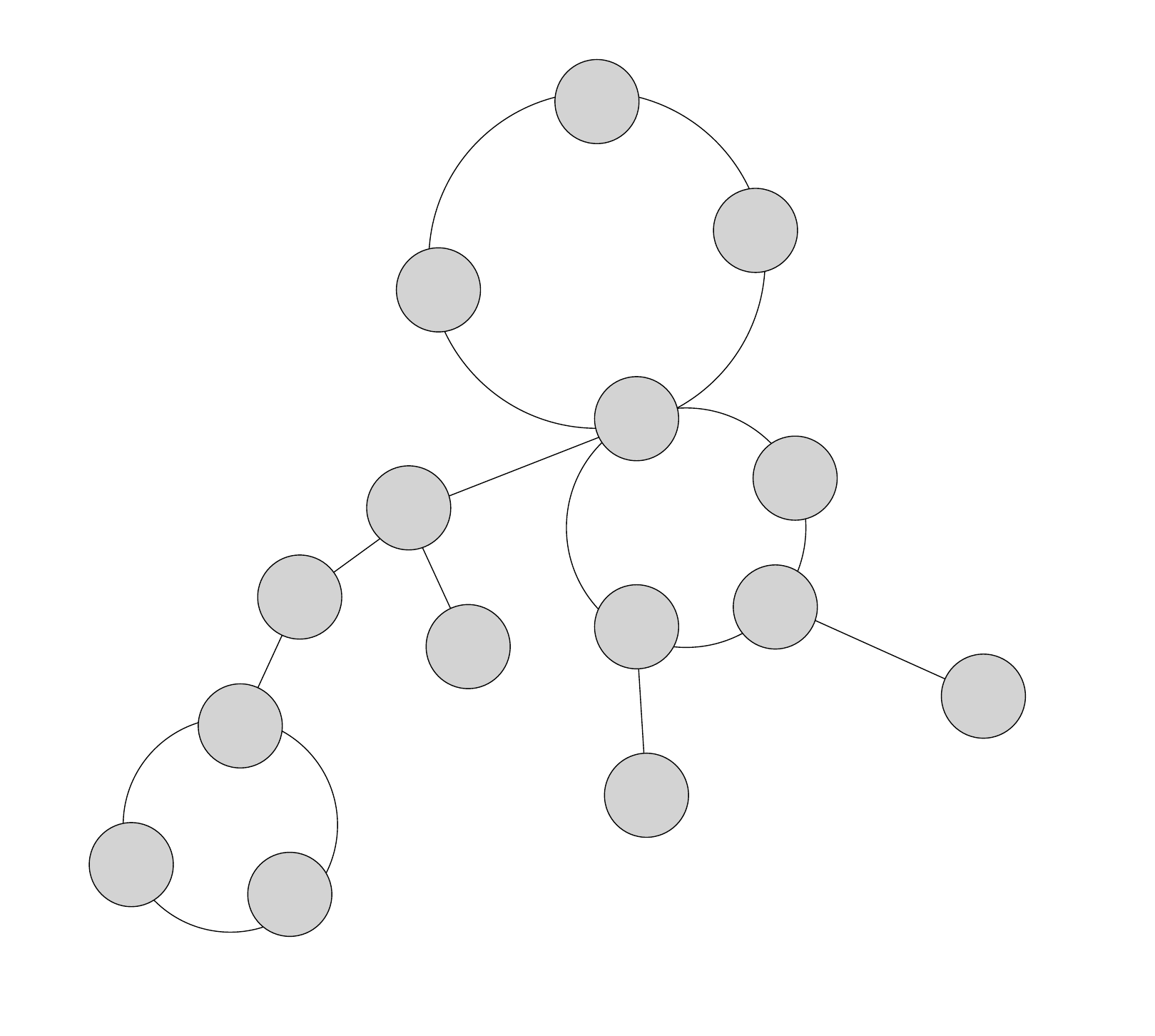}}
	\hfill
	\subfigure[$\mathcal{T}$]{\label{fi:cactus-b}\includegraphics[height=5cm]{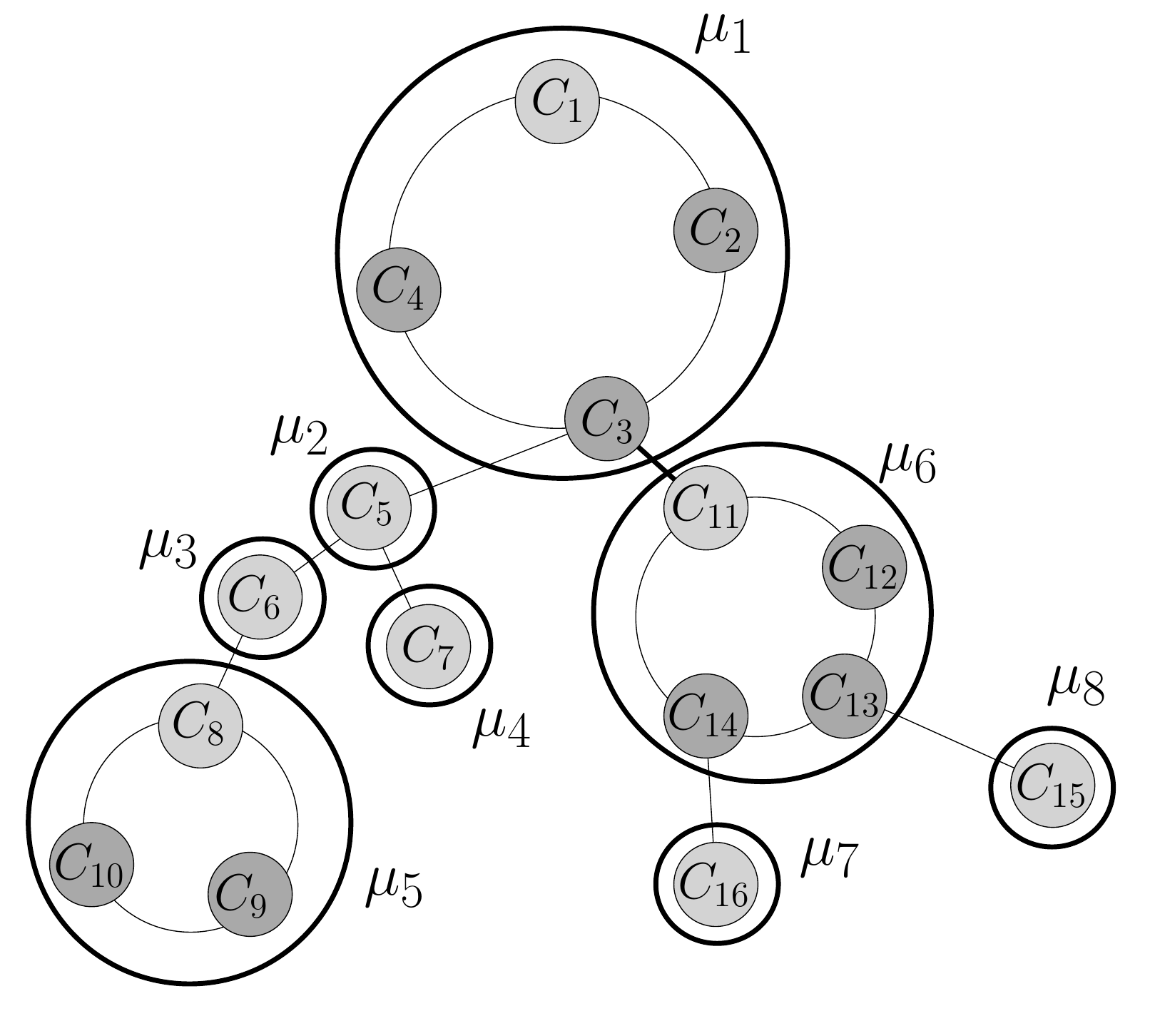}}
	\caption{%
	(a) An intersection graph $\Gcross$ that is a cactus.
	(b) The decomposition tree $\mathcal{T}$ of $\Gcross$. Clusters $C_{11}$ of $\mu_6$ and $C_3$ of $\mu_1$ correspond to the same vertex of $\Gcross$.}
\end{figure}

\newcommand{\cactusalgorithm}{Let  $\langle G,\phi \rangle$ be an \FM such that $\Gcross$ is a cactus. There exists a polynomial-time algorithm that tests whether $\langle G, \phi \rangle$ has a planar 0-bend convex-hull drawing.}
\wormhole{cactusalgorithm}
\begin{theorem}\label{thm:cacti}
	\cactusalgorithm
\end{theorem}
\begin{proof}
By Lemmas~\ref{lem:path} and~\ref{lem:cycle}, the statement holds when $\Gcross$ is a path or a cycle. In the general case, our testing algorithm decomposes $\Gcross$ into its tree $\mathcal{T}$ (as in Fig.~\ref{fi:cactus-b}), roots $\mathcal{T}$ at any node, and visits $\mathcal{T}$ bottom-up. More precisely, each vertex of $\Gcross$ corresponds to a convex hull $CH(u)$ of a mobile vertex $u$, and it has a one-to-one correspondence with a cluster $C(u)$ of $\Gcell$. Thus, each node $\mu$ of $\mathcal{T}$ corresponds to either a single cluster of $\Gcell$ or to a cycle of clusters of $\Gcell$.
Note that when in $\Gcross$ two cycles share a vertex (cluster of $\Gcell$), we replicate such a vertex in both nodes of $\mathcal{T}$ that correspond to the two cycles.
For example, in Fig.~\ref{fi:cactus-b} cluster $C_{11}$ inside $\mu_6$ and cluster $C_3$ inside $\mu_1$ correspond to the same vertex of $\Gcross$.
Once the root of $\mathcal{T}$ has been chosen, we define the \emph{anchor} of $\mu$ as the cluster that connects $\mu$ to its parent node in $\mathcal{T}$ (light-gray in Fig.~\ref{fi:cactus-b}).
During the bottom-up visit of $\mathcal{T}$, two cases are possible when a node~$\mu$ is visited:

\begin{itemize}
\item {\bf\boldmath $\mu$ is a leaf}. If $\mu$ contains a single cluster (i.e., its anchor), then all its cell vertices are marked as active; if $\mu$ contains a cycle of clusters, then the active cell vertices of its anchor are computed as for $C(u_1)$ in the proof of Lemma~\ref{lem:cycle}.

\item {\bf\boldmath $\mu$ is an internal node}. Let $\nu_1, \nu_2, \dots, \nu_k$ be the children of $\mu$ in $\mathcal{T}$ and denote by $C_{q_i}$ the cluster of $\mu$ connected to the anchor of $\nu_i$. Note that $C_{q_i}$ may coincide with some $C_{q_j}$ if $q_i \neq q_j$. Also, the anchor of $\nu_i$ and $C_{q_i}$ may correspond to the same vertex of $\Gcross$.

\begin{itemize}
\item For each $i=1, \dots, k$, if the anchor of $\nu_i$ differs from $C_{q_i}$ in $\Gcross$, remove from $C_{q_i}$ all cell vertices that are not connected to an active cell vertex of the anchor of $\nu_i$ in $\Gcell$, as they cannot occur in any skeleton of $\Gcell$. On the other hand, if the anchor of $\nu_i$ and $C_{q_i}$ coincide in $\Gcross$, remove from $C_{q_i}$ all vertices of the anchor of $\nu_i$ that are not marked as active.

\item Now, if $\mu$ contains a single cluster (i.e., its anchor), then all its remaining cell vertices are marked as active; if $\mu$ contains a cycle of clusters, then the active cell vertices of its anchor are computed as in Lemma~\ref{lem:cycle}. At this point, if the anchor of $\mu$ contains no active vertex, the algorithm stops and the instance is rejected, as a skeleton does not exist.
\end{itemize}
\end{itemize}

Once the bottom-up visit of $\mathcal{T}$ ends, the test is positive if and only if the anchor of the root node of $\mathcal{T}$ has an active cell vertex $w$, and in this case one can reconstruct a skeleton of $\Gcell$ starting from $w$ and visiting $\mathcal{T}$ top-down. In particular, during the top-down visit, for each node $\mu$ of $\mathcal{T}$, any active vertex in the anchor of $\mu$ can be arbitrarily selected, as it is connected to the parent node of $\mu$ by construction. Also, if $\mu$ corresponds to a simple cycle of clusters, the construction of a cycle that connects these clusters is done as in Lemma~\ref{lem:cycle}.

Concerning the time complexity, the above algorithm takes polynomial time in the size of $\Gcell$. Indeed, the number of clusters that may occur in multiple nodes of $\mathcal{T}$ (i.e., those that are shared by multiple cycles of clusters) is at most the number of cycles in $\Gcross$. Therefore, the total number of clusters over all nodes of $\mathcal{T}$ is linear in the number of clusters of $\Gcell$. This also implies that the total number of cell vertices over all clusters of $\mathcal{T}$ is linear in the number of cell vertices in $\Gcell$. Finally, each node $\mu$ of $\mathcal{T}$ is visited twice (once in the bottom-up visit and once in the top-down visit), and in each visit of $\mu$ the algorithm has a running time that is polynomial in the number of cell vertices in the clusters of $\mu$.
\end{proof}

Finally, we show the following \NP-completeness result on a combinatorial generalization of our problem. We re-use the terminology of the \FMbend{0} problem to emphasize the analogies.

\newcommand{\skeletonhardness}{Let $\Gcross = (\mathcal C, \mathcal E)$ be a graph, where $\mathcal C$ is a set of disjoint clusters of cells. Also, let $\Gcell = (V,E)$ be a graph, where each $v \in V$ is a cell of a cluster $C(v) \in \mathcal C$ and $(u,v) \in E$ only if $(C(u),C(v)) \in \mathcal E$. It is \NP-complete to test if there is a subset $V' \subseteq V$ of skeleton vertices, containing exactly one cell from each cluster in $\mathcal C$ such that the induced subgraph $\Gcell[V']$ is isomorphic to~$\Gcross$.}
\wormhole{skeletonhardness}
\begin{theorem}\label{thm:skeletonhardness}
\skeletonhardness
\end{theorem}
\begin{proof}[sketch]
The problem is clearly in $\NP$. The hardness proof is by reduction from \textsc{3Sat}. For a boolean \textsc{3Sat} formula $\psi$  create a cluster $C(x)$ for each variable $x$ in $\psi$ and a cluster $C(\gamma)$ for each clause $\gamma$ of $\psi$. In $\Gcross$ each clause cluster is adjacent to the three clusters of the variables occurring in the clause.  Each variable cluster $C(x)$ consists of two cells in $\Gcell$, one for the positive literal $x$ and one for its negation $\neg{x}$. Also, each clause cluster $C(\gamma)$ contains three cells, one for each literal. Finally, connect each literal cell of a clause $\lambda$ to the corresponding cell of its variable cluster and to all four cells of the other two variables of $\gamma$. It can be seen that $\psi$ has a satisfying truth assignment iff there exists a subset of skeleton vertices in $\Gcell$ that induces a subgraph isomorphic to $\Gcross$.
\end{proof}

\section{1-bend Drawings in the $h$-Strip Drawing Model}\label{se:1-bend}
Our model for \FMbends{1} is inspired by the boundary labeling approach~\cite{DBLP:journals/jgaa/BekosCFH0NRS15}, where mobile vertices are regarded as labels that must be connected to the fixed vertices. In the boundary labeling model, the fixed vertices are inside a single rectangular region and each label is either to the left or to the right of this region. Our model allows for multiple rectangular regions (corresponding to horizontal strips); each mobile vertex is placed outside of these regions, either below or above each of them. To avoid long edges and make the drawing more readable, edges are not allowed to traverse regions. 

More formally, our model is called the \emph{$h$-strip model} and is defined as follows. Let $\langle G=(V_f,V_m,E),\phi \rangle$ be an \FM and assume that the vertices of $V_f$ all have distinct $x$-coordinates (this condition is always achievable by a suitable rotation of the plane). For the sake of simplicity, for a vertex $u \in V_f$, we do not distinguish between $u$ and its fixed position $\phi(u)$. Let $\mathcal{S} = \{S_1, S_2, \dots, S_h\}$ $(h \geq 1)$ be a top-to-bottom sequence of (closed) disjoint horizontal strips of the plane that partition $V_f$, i.e., each $S_i$ has a finite height and infinite width, each vertex of $V_f$ lies in one $S_i$, and $S_i \cap S_{i+1} = \emptyset$ for $i = 1, \dots, h-1$. Since the strips are disjoint, there is always a non-empty region of the plane between two consecutive strips, which does not contain fixed vertices. Also, there are no fixed vertex above $S_1$ and below $S_h$. Any point that is not inside a strip is called a \emph{free point}. For a vertex $u \in V_f$, the strip that contains $u$ is called the \emph{strip of $u$}.

A \emph{$1$-bend drawing of $G$ within $\mathcal{S}$} is defined as follows; see Fig.~\ref{fi:strip-model}: $(i)$ Each vertex $v \in V_m$ is mapped to a distinct free point. $(ii)$ Each edge $e=(u,v)$, with $u \in V_f$ and $v \in V_m$ consists of a segment $\overline{vp}$ from $v$ to a point $p$ on the boundary of the strip of $u$ and of a vertical segment $\overline{pu}$; all points of $\overline{vp}$ but $p$ are free points, while $\overline{pu}$ is completely inside the strip of $u$. $(iii)$ No edge intersects the boundary of a strip twice and no two edges cross in a free point.

\begin{wrapfigure}[16]{R}{0.40\textwidth}
\begin{center}
\includegraphics[scale=0.30]{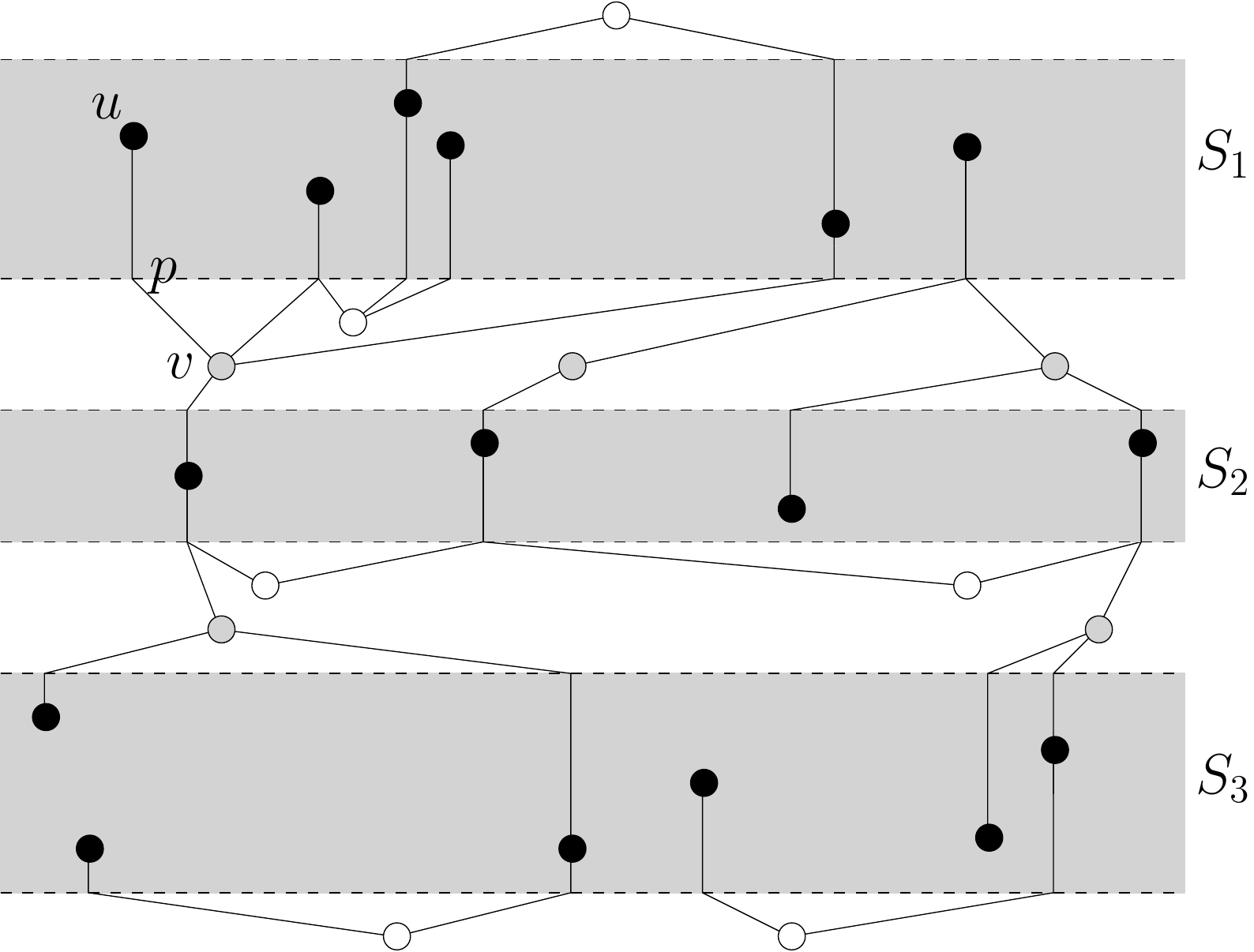}
\caption{A $3$-strip drawing.}\label{fi:strip-model}
\end{center}
\end{wrapfigure}

Note that in the $h$-strip model two distinct edges $e_1=(u,v_1)$ and $e_2=(u,v_2)$, where $u \in V_f$, share their vertical segments if these segments are incident to $u$ both from below or both from above. This overlap does not create ambiguity and reduces the visual complexity caused by the edges. Figure~\ref{fi:strip-model} shows a $1$-bend drawing of a bigraph within a given set of three strips (gray regions), with fixed vertices in black. Also note that, if an \FM $G$ has no $1$-bend drawing for a set $\mathcal{S}$ of strips, splitting an element of $\mathcal{S}$ into two strips may lead to a feasible solution; see Fig.~\ref{fi:splitting-a}. Conversely, splitting a strip may transform a positive instance into a negative one; see Fig.~\ref{fi:splitting-b}. We prove the following.

\newcommand{\onebendstripmodel}{Let $\langle G=(V_f,V_m,E),\phi \rangle$ be an $n$-vertex \FM and let $\mathcal{S}$ be a set of horizontal strips that partition $V_f$. There exists an $O(n)$-time algorithm that tests whether $\langle G, \phi \rangle$ admits a $1$-bend drawing within $\mathcal{S}$.}

\wormhole{onebendstripmodel}
\begin{theorem}\label{th:1-bend-strip-model}
	\onebendstripmodel
\end{theorem}
\begin{proof}[sketch]
Call a fixed vertex \emph{black}, a mobile vertex with all neighbors in the same strip \emph{white}, and the remaining vertices \emph{gray}. A gray vertex with neighbors in two consecutive strips must lie between them, while each white vertex can lie either above or below the strip of its neighbors. If a grey vertex has neighbors that are not in two consecutive strips, the instance is immediately rejected.

Let $\mathcal{S} = \{S_1, \dots, S_h\}$ be the sequence of strips. For each strip $S_i \in \mathcal{S}$, let $V_f^i = \{u_1^i, \dots, u_{r_i}^i\}$ be the left-to-right sequence of black vertices inside $S_i$. Also, let $V_m^i$ be the set of mobile vertices connected to some vertex of $V_f^i$. Arranging the vertices of $V_m^i$ above or below $S_i$ so to avoid crossings between their incident edges equals to assigning each of them either above or below the half-plane determined by a horizontal line that contains $u_1^i, \dots, u_{r_i}^i$, in this left-to-right order. Hence, testing if a $1$-bend drawing within $\mathcal{S}$ exists generalizes testing the existence of a $0$-bend drawing when all fixed vertices are collinear. As in Theorem~\ref{th:collinear}, this problem is reduced to testing planarity of a graph $G'$ suitably defined by augmenting $G$. Namely, for each $S_i$, add a cycle $C_i$ connecting all edges of $V_f^i$ in their left-to-right order; then, subdivide edge $(u_1^i, u_{r_i}^i)$ of $C_i$ with three vertices $v_1^i, v_2^i, v_3^i$, in this order from $u_1^i$ to $u_{r_i}^i$, and call $C'_i$ the subdivision of $C_i$; finally, for each $i = 1, \dots, h-1$ and $j = 1,2,3$, connect $v_j^i$ to $v_j^{i+1}$. Graph $G'$ is planar iff it has a planar embedding where $C'_i$ is inside $C'_{i+1}$ ($i \in \{1, \dots, h-1\}$). A mobile vertex $w$ between $C'_i$ and $C'_{i+1}$ corresponds to placing $w$ above $S_{i+1}$ and below $S_i$. If $w$ is in the outer face of $G'$ then $w$ is below $S_h$, and if $w$ is inside $C'_1$ then $w$ is above $S_1$. Since the size of $G'$ is linear in the size of $G$ and graph planarity testing is linear-time solvable, the statement holds.
\end{proof}

\begin{figure}[tb]
\centering
\subfigure[]{\label{fi:splitting-a}\includegraphics[scale=0.43]{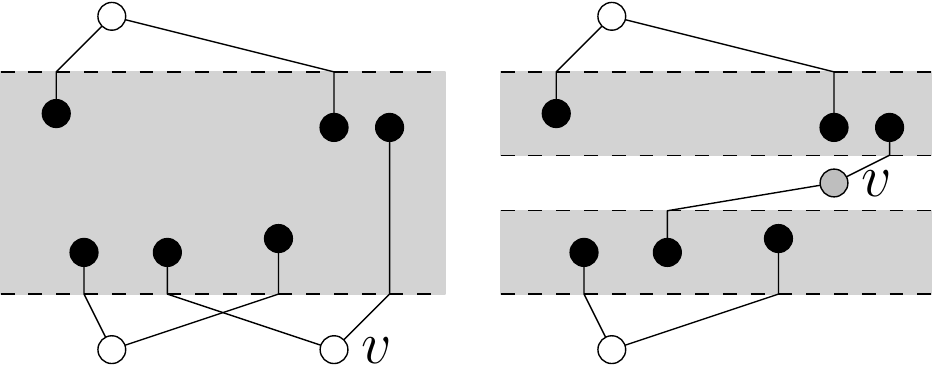}}
\hspace{3cm}
\subfigure[]{\label{fi:splitting-b}\includegraphics[scale=0.43]{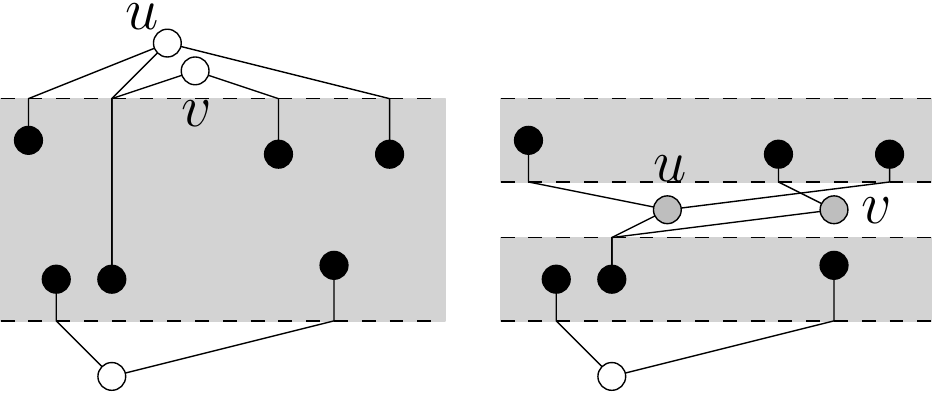}}
\caption{(a) An instance with a single strip and no solution (left); splitting the strip into two strips, the instance becomes feasible (right). (b) A positive instance with a single strip; splitting the strip into two strips, it becomes unfeasible.}
\end{figure}

The next result immediately follows by iterating the technique in the proof of Theorem~\ref{th:1-bend-strip-model} over all possible ways of partitioning $V_f$ into $h$ strips.

\begin{corollary}\label{co:1-bend-strip-model}
Let $\langle G=(V_f,V_m,E),\phi \rangle$ be an $n$-vertex \FM and let $h \in \mathbb N^+$ be a constant. There is an $O(|V_f|^{h-1}n)$-time algorithm that tests if $\langle G, \phi \rangle$ has a $1$-bend drawing within $\mathcal{S}$, for some set $\mathcal{S}$ of $h$ strips that partition $V_f$.
\end{corollary}

\section{Conclusions and Open Problems}
We introduced \FMs, showed that the \FMbend{k} problem is \NP-hard in the general case, and gave polynomial-time algorithms for $k \in \{0,1\}$ in some interesting restricted cases. Several open research questions remain: 

\smallskip\noindent{\bf Q1.} We could solve the \FMbend{0} problem for convex-hull drawings if the CH intersection graph is a cactus and show that it is \NP-complete in a non-geometric setting. Can we solve the problem for larger classes of convex-hull drawings in polynomial time or extend the \NP-completeness to our geometric setting? More generally, for which other layout constraints or sub-families of \FMs does the \FMbend{k} problem become tractable?

\smallskip\noindent{\bf Q2.} Our focus was on proving the existence of polynomial-time algorithms under certain layout constraints, but some of the algorithms have high time complexity. Thus, finding more efficient algorithms is of interest.

\smallskip\noindent{\bf Q3.} We focused on crossing-free drawings of \FMs. Relaxing the planarity requirement (e.g., for a given maximum number of permitted crossings per edge) is an interesting variant, as well as, designing heuristics or exact approaches for crossing/bend minimization.

\bibliography{FMBigraphs}

\begin{thebibliography}{10}
\providecommand{\url}[1]{\texttt{#1}}
\providecommand{\urlprefix}{URL }

\bibitem{DBLP:journals/tcs/BadentGL08}
Badent, M., {Di Giacomo}, E., Liotta, G.: Drawing colored graphs on colored
  points. Theoretical Computer Science  408(2-3),  129--142 (2008)

\bibitem{bgnn-rbl-15}
Barth, L., Gemsa, A., Niedermann, B., N{\"{o}}llenburg, M.: On the readability
  of boundary labeling. In: Di~Giacomo, E., Lubiw, A. (eds.) Graph Drawing and
  Network Visualization (GD'15). LNCS, vol. 9411, pp. 515--527. Springer
  International Publishing (2015)

\bibitem{DBLP:journals/jgaa/BekosCFH0NRS15}
Bekos, M.A., Cornelsen, S., Fink, M., Hong, S., Kaufmann, M., N{\"{o}}llenburg,
  M., Rutter, I., Symvonis, A.: Many-to-one boundary labeling with backbones.
  Journal of Graph Algorithms and Applications  19(3),  779--816 (2015)

\bibitem{DBLP:journals/comgeo/BekosKSW07}
Bekos, M.A., Kaufmann, M., Symvonis, A., Wolff, A.: Boundary labeling: Models
  and efficient algorithms for rectangular maps. Computational Geometry  36(3),
   215--236 (2007)

\bibitem{DBLP:conf/compgeom/Biedl98}
Biedl, T.C.: Drawing planar partitions {I:} {LL}-drawings and {LH}-drawings.
  In: Janardan, R. (ed.) Computational Geometry (SoCG'98). pp. 287--296. {ACM}
  (1998)

\bibitem{DBLP:conf/wg/BiedlKM98}
Biedl, T.C., Kaufmann, M., Mutzel, P.: Drawing planar partitions {II:}
  {HH}-drawings. In: Hromkovic, J., S{\'{y}}kora, O. (eds.) Graph-Theoretic
  Concepts in Computer Science (WG'98). LNCS, vol. 1517, pp. 124--136. Springer
  (1998)

\bibitem{DBLP:journals/jcss/BoothL76}
Booth, K.S., Lueker, G.S.: Testing for the consecutive ones property, interval
  graphs, and graph planarity using pq-tree algorithms. Journal of Computer and
  System Sciences  13(3),  335--379 (1976)

\bibitem{DBLP:journals/algorithmica/BrandesEEFFGGHKKa11}
Brandes, U., Erten, C., Estrella{-}Balderrama, A., Fowler, J.J., Frati, F.,
  Geyer, M., Gutwenger, C., Hong, S., Kaufmann, M., Kobourov, S.G., Liotta, G.,
  Mutzel, P., Symvonis, A.: Colored simultaneous geometric embeddings and
  universal pointsets. Algorithmica  60(3),  569--592 (2011)

\bibitem{cegl-dgppo-12}
Chambers, E.W., Eppstein, D., Goodrich, M.T., L\"offler, M.: Drawing graphs in
  the plane with a prescribed outer face and polynomial area. Journal of Graph
  Algorithms and Applications  16(2),  243--259 (2012)

\bibitem{DBLP:books/ph/BattistaETT99}
{Di Battista}, G., Eades, P., Tamassia, R., Tollis, I.G.: Graph Drawing:
  Algorithms for the Visualization of Graphs. Prentice-Hall (1999)

\bibitem{DBLP:journals/jgaa/GiacomoDLMTW08}
{Di Giacomo}, E., Didimo, W., Liotta, G., Meijer, H., Trotta, F., Wismath,
  S.K.: k-colored point-set embeddability of outerplanar graphs. Journal of
  Graph Algorithms and Applications  12(1),  29--49 (2008)

\bibitem{DBLP:journals/algorithmica/GiacomoLT10}
{Di Giacomo}, E., Liotta, G., Trotta, F.: Drawing colored graphs with
  constrained vertex positions and few bends per edge. Algorithmica  57(4),
  796--818 (2010)

\bibitem{dgk-pdhgg-11}
Duncan, C.A., Goodrich, M.T., Kobourov, S.G.: Planar drawings of higher-genus
  graphs. Journal of Graph Algorithms and Applications  15(1),  7--32 (2011)

\bibitem{DBLP:conf/ciac/FossmeierK97}
F{\"{o}}{\ss}meier, U., Kaufmann, M.: Nice drawings for planar bipartite
  graphs. In: Bongiovanni, G.C., Bovet, D.P., {Di Battista}, G. (eds.)
  Algorithms and Complexity (CIAC'97). LNCS, vol. 1203, pp. 122--134. Springer
  (1997)

\bibitem{DBLP:journals/dcg/GoaocKOSSW09}
Goaoc, X., Kratochv{\'{\i}}l, J., Okamoto, Y., Shin, C., Spillner, A., Wolff,
  A.: Untangling a planar graph. Discrete {\&} Computational Geometry  42(4),
  542--569 (2009)

\bibitem{h-a-04}
Halperin, D.: Arrangements. In: Goodman, J.E., O'Rourke, J. (eds.) Handbook of
  Discrete and Computational Geometry, chap.~24, pp. 529--562. CRC Press LLC,
  Boca Raton, FL (2004)

\bibitem{h-gt-72}
Harary, F.: Graph Theory. Addison-Wesley, Reading, MA (1972)

\bibitem{HongN08}
Hong, S.H., Nagamochi, H.: Convex drawings of graphs with non-convex boundary
  constraints. Discrete Applied Mathematics  156(12),  2368--2380 (2008)

\bibitem{DBLP:journals/jacm/HopcroftT74}
Hopcroft, J.E., Tarjan, R.E.: Efficient planarity testing. Journal of the {ACM}
   21(4),  549--568 (1974)

\bibitem{DBLP:conf/hci/ItoMT09}
Ito, T., Misue, K., Tanaka, J.: Sphere anchored map: {A} visualization
  technique for bipartite graphs in {3D}. In: Jacko, J.A. (ed.) Human-Computer
  Interaction (HCI'09). LNCS, vol. 5611, pp. 811--820. Springer (2009)

\bibitem{DBLP:conf/dagstuhl/1999dg}
Kaufmann, M., Wagner, D. (eds.): Drawing Graphs, Methods and Models, LNCS, vol.
  2025. Springer (2001)

\bibitem{DBLP:journals/jgaa/KaufmannW02}
Kaufmann, M., Wiese, R.: Embedding vertices at points: Few bends suffice for
  planar graphs. Journal of Graph Algorithms and Applications  6(1),  115--129
  (2002)

\bibitem{DBLP:journals/algorithmica/KindermannNRS0W16}
Kindermann, P., Niedermann, B., Rutter, I., Schaefer, M., Schulz, A., Wolff,
  A.: Multi-sided boundary labeling. Algorithmica  76(1),  225--258 (2016)

\bibitem{l-cmblwh-10}
Lin, C.: Crossing-free many-to-one boundary labeling with hyperleaders. In:
  {IEEE} Pacific Visualization Symposium PacificVis 2010, Taipei, Taiwan, March
  2-5, 2010. pp. 185--192. {IEEE} Computer Society (2010)

\bibitem{MchedlidzeNR16}
Mchedlidze, T., N{\"{o}}llenburg, M., Rutter, I.: Extending convex partial
  drawings of graphs. Algorithmica  76(1),  47--67 (2016)

\bibitem{McKee1672910}
{McKee}, T.A., {McMorris}, F.R.: Topics in Intersection Graph Theory. SIAM
  Monographs on Discrete Mathematics and Applications (1999)

\bibitem{DBLP:journals/ieicet/Misue08}
Misue, K.: Anchored map: Graph drawing technique to support network mining.
  {IEICE} Transactions  91-D(11),  2599--2606 (2008)

\bibitem{DBLP:conf/iv/MisueZ11}
Misue, K., Zhou, Q.: Drawing semi-bipartite graphs in anchor+matrix style. In:
  Information Visualisation (IV'11), London, UK, 13-15 July 2011. pp. 26--31.
  {IEEE} Computer Society (2011)

\bibitem{DBLP:conf/dagstuhl/Neyer99}
Neyer, G.: Map labeling with application to graph drawing. In: Kaufmann, M.,
  Wagner, D. (eds.) Drawing Graphs, Methods and Models. LNCS, vol. 2025, pp.
  247--273. Springer (1999)

\bibitem{DBLP:journals/gc/PachW01}
Pach, J., Wenger, R.: Embedding planar graphs at fixed vertex locations. Graphs
  and Combinatorics  17(4),  717--728 (2001)

\bibitem{DBLP:journals/ijfcs/Patrignani06}
Patrignani, M.: On extending a partial straight-line drawing. International
  Journal of Foundations of Computer Science  17(5),  1061--1070 (2006)

\bibitem{Patrignani06}
Patrignani, M.: On extending a partial straight-line drawing. In: Healy, P.,
  Nikolov, N.S. (eds.) Graph Drawing 2005, LNCS, vol. 3843, pp. 380--385.
  Springer (2006)

\bibitem{DBLP:journals/vlc/Purchase02}
Purchase, H.C.: Metrics for graph drawing aesthetics. Journal of Visual
  Languages and Computing  13(5),  501--516 (2002)

\bibitem{DBLP:journals/ese/PurchaseCA02}
Purchase, H.C., Carrington, D.A., Allder, J.: Empirical evaluation of
  aesthetics-based graph layout. Empirical Software Engineering  7(3),
  233--255 (2002)

\bibitem{DBLP:reference/crc/2013gd}
Tamassia, R. (ed.): Handbook on Graph Drawing and Visualization. Chapman and
  Hall/CRC (2013)

\bibitem{DBLP:reference/cg/TamassiaL04}
Tamassia, R., Liotta, G.: Graph drawing. In: Goodman, J.E., O'Rourke, J. (eds.)
  Handbook of Discrete and Computational Geometry, Second Edition., pp.
  1163--1185. Chapman and Hall/CRC (2004)

\bibitem{t-hdg-63}
Tutte, W.T.: How to draw a graph. Proceedings of the London Mathematical
  Society  13(3),  743--768 (1963)

\bibitem{Wolff-96}
Wolff, A., Strijk, T.: The map-labeling bibliography (1996), {URL}:
  http://i11www.ira.uka.de/map-labeling/bibliography.

\end{thebibliography}
\bibliographystyle{splncs03}

\newpage
\appendix
\section*{Appendix}

\begin{backInTime}{collinear}
\begin{theorem}
	\collinear
\end{theorem}
\begin{proof}
Assume w.l.o.g.\ that all vertices of $V_f$ lie on a horizontal line $\ell$. In any planar 0-bend drawing of $\langle G, \phi \rangle$, we can assume that every vertex $w \in V_m$ is either above or below $\ell$. Indeed, if $w$ lies on $\ell$ then $w$ has degree either one or two, and in this last case it is incident to two consecutive vertices of $V_f$ along $\ell$: We can always slightly move $w$ above or below $\ell$ without changing the planar embedding of the drawing. Hence, deciding whether $\langle G, \phi \rangle$ has a planar 0-bend drawing is equivalent to deciding whether there exists an assignment of each mobile vertex to one of the two half planes determined by $\ell$ that avoids edge crossings. This problem coincides with testing the planarity of a graph $G'$ obtained by augmenting $G$ with a cycle that connects all fixed vertices in the order they appear along $\ell$: A vertex inside (outisde of) the cycle corresponds to a vertex above (below) $\ell$. Since the size of $G'$ is linear in the size of $G$ and since the graph planarity testing problem is linear-time solvable~\cite{DBLP:journals/jcss/BoothL76,DBLP:journals/jacm/HopcroftT74}, the statement follows.
\end{proof}
\end{backInTime}

\begin{backInTime}{skeletonhardness}
\begin{theorem}
	\skeletonhardness
\end{theorem}
\begin{proof}
The problem is clearly in \NP. The proof of the hardness is by reduction from \textsc{3Sat}. For a boolean \textsc{3Sat} formula $\psi$  create a cluster $C(x)$ for each variable $x$ in $\psi$ and a cluster $C(\gamma)$ for each clause $\gamma$ of $\psi$. In $\Gcross$ each clause cluster is adjacent to the three clusters of the variables occurring in the clause.  Each variable cluster $C(x)$ consists of two cells in $\Gcell$, one for the positive literal $x$ and one for its negation $\neg{x}$. Also, each clause cluster $C(\gamma)$ contains three cells, one for each literal. Finally, connect each literal cell of a clause $\lambda$ to the corresponding cell of its variable cluster and to all four cells of the other two variables of $\gamma$.

We now show that $\psi$ has a satisfying truth assignment iff there exists a subset of skeleton vertices in $\Gcell$ that induces a subgraph isomorphic to $\Gcross$. Assume that we know a satisfying variable assignment of $\psi$. We select the ``\emph{true}'' cells of the variable clusters and one satisfied literal of each clause. The induced subgraph of this set of cells is isomorphic to $\Gcross$ as the satisfied literal cell of each clause $\gamma$ covers all three edges of $C(\gamma)$ to its three adjacent variable clusters. Conversely, if we have a subset of skeleton vertices of $\Gcell$ that induce a subgraph isomorphic to $\Gcell$, then setting the literals of the set of selected literal cells to \emph{true} satisfies $\psi$. Otherwise, in an unsatisfied clause, none of the three clause cells would connect to the selected cells of all three adjacent variable clusters, which contradicts the skeleton property.
\end{proof}
\end{backInTime}

\begin{figure}[t]
\centering
\subfigure[]{\label{fi:strip-model-proof-a}\includegraphics[scale=0.4]{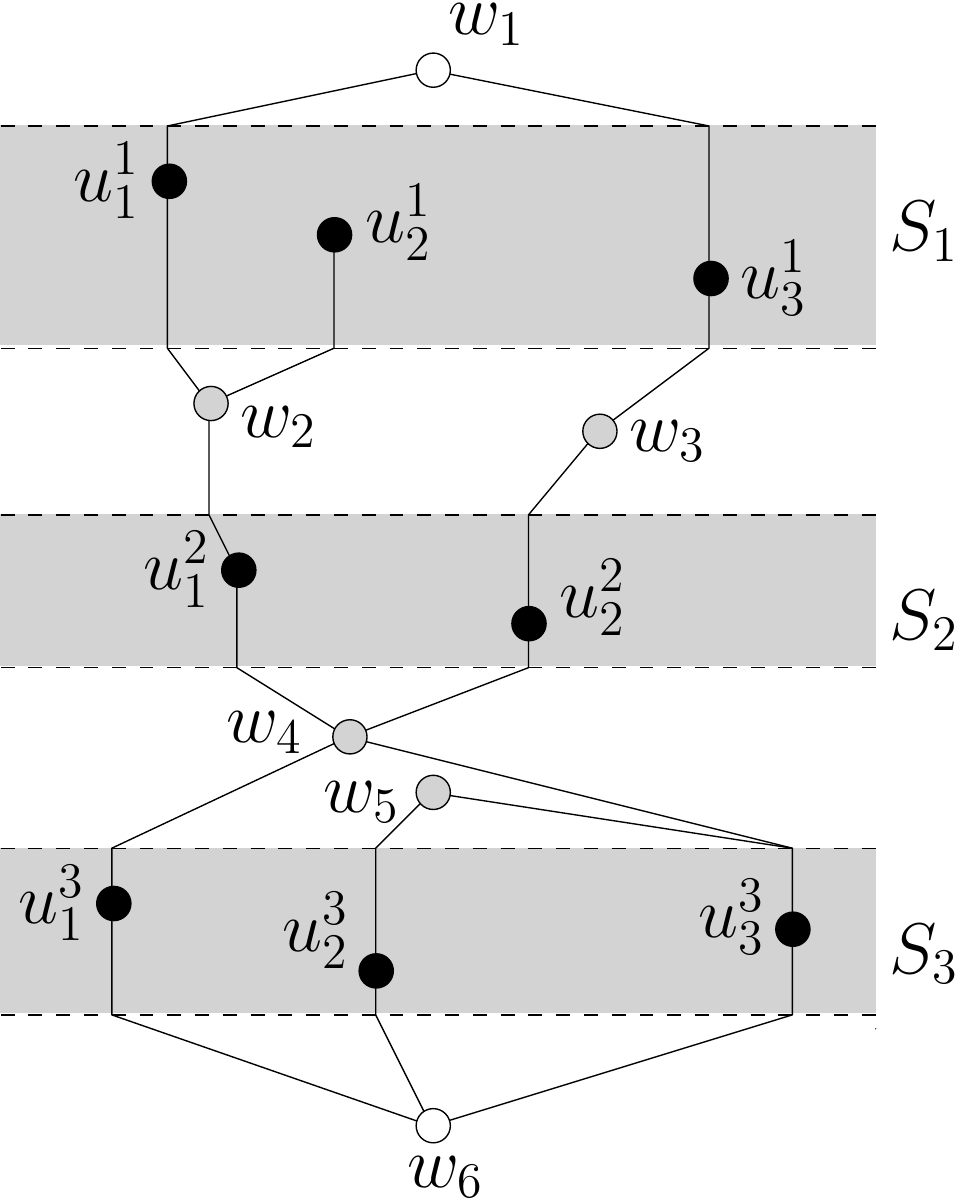}}
\hfill
\subfigure[]{\label{fi:strip-model-proof-b}\includegraphics[scale=0.45]{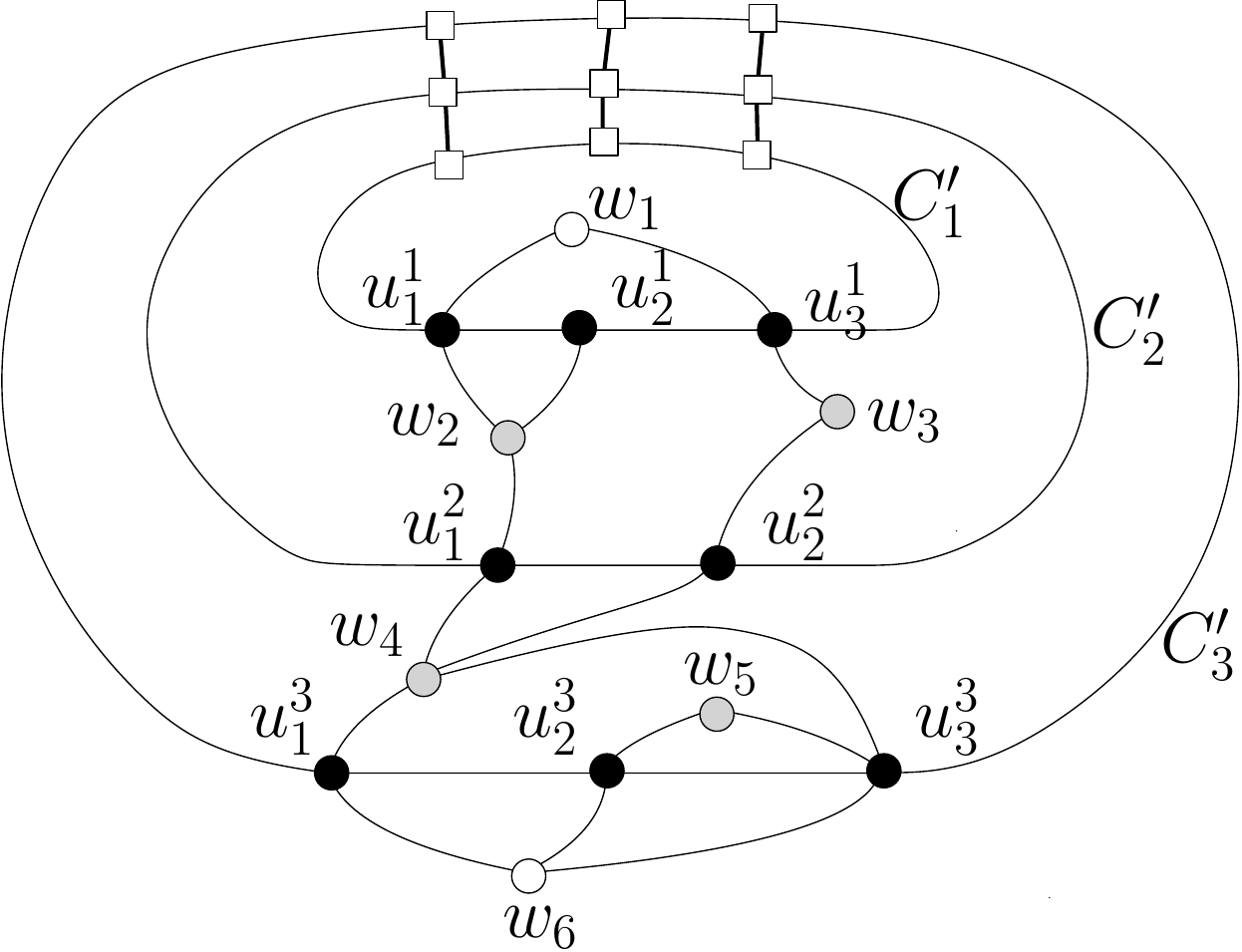}}
\caption{Illustration for the proof of Theorem~\ref{th:1-bend-strip-model}: (a) A planar 1-bend drawing in the strip model; (b) The corresponding planar embedding of the graph $G'$.}\label{fi:strip-model-proof}
\end{figure}

\begin{backInTime}{onebendstripmodel}
\begin{theorem}\label{th:1-bend-strip-model}
	\onebendstripmodel
\end{theorem}
\begin{proof}
Call \emph{black} a fixed vertex, \emph{white} a mobile vertex with all neighbors in the same strip, and \emph{gray} the remaining vertices. A gray vertex with neighbors in two consecutive strips must lie between them, while each white vertex can lie either above or below the strip of its (black) neighbors. If a grey vertex has neighbors that are not in two consecutive strips, the instance is immediately rejected.

Let $\mathcal{S} = \{S_1, \dots, S_h\}$ be the sequence of strips. For each strip $S_i \in \mathcal{S}$, let $V_f^i = \{u_1^i, \dots, u_{r_i}^i\}$ be the left-to-right sequence of black vertices inside $S_i$. Also, let $V_m^i$ be the set of mobile vertices connected to some vertex of $V_f^i$. Arranging the vertices of $V_m^i$ above or below $S_i$ so to avoid crossings between their incident edges equals to assigning each of them either above or below the half-plane determined by a horizontal line that contains $u_1^i, \dots, u_{r_i}^i$, in this left-to-right order. Hence, testing if a $1$-bend drawing within $\mathcal{S}$ exists generalizes testing the existence of a $0$-bend drawing when all fixed vertices are collinear. As in Theorem~\ref{th:collinear}, this problem is reduced to testing planarity of a graph $G'$ suitably defined by augmenting $G$. Namely, for each $S_i$, add a cycle $C_i$ connecting all edges of $V_f^i$ in their left-to-right order; then, subdivide edge $(u_1^i, u_{r_i}^i)$ of $C_i$ with three dummy vertices $v_1^i, v_2^i, v_3^i$, in this order from $u_1^i$ to $u_{r_i}^i$, and call $C'_i$ the subdivision of $C_i$; finally, for each $i = 1, \dots, h-1$ and $j = 1,2,3$, connect $v_j^i$ to $v_j^{i+1}$. See Fig.~\ref{fi:strip-model-proof} for an illustration; in the figure, the dummy vertices $v_j^i$ are represented by small squares.

Denote by $G''$ the graph obtained from $G'$ by removing all mobile vertices. It is immediate to see that $G'$ is planar if and only if it has a planar embedding such that $C'_h$ coincides with the external face of $G''$. In such an embedding for $G'$, the edges $(v_j^i, v_j^{i+1})$ force all cycles to be nested in such a way that $C'_i$ is inside $C'_{i+1}$ ($i \in \{1, \dots, h-1\}$). A mobile vertex $w$ between $C'_i$ and $C'_{i+1}$ corresponds to placing $w$ above $S_{i+1}$ and below $S_i$. If $w$ is in the outer face of $G'$ then $w$ is below $S_h$, and if $w$ is inside $C'_1$ then $w$ is above $S_1$. Since the size of $G'$ is linear in the size of $G$ and graph planarity testing is linear-time solvable, the statement holds.
\end{proof}
\end{backInTime}
\end{document}